\documentclass[lettersize,journal]{IEEEtran}
\usepackage{amsthm}
\usepackage{amsmath}
\usepackage{amsmath,amsfonts}
\usepackage{algorithmic}
\usepackage{algorithm}
\usepackage{array}
\usepackage[caption=false,font=normalsize,labelfont=sf,textfont=sf]{subfig}
\usepackage{textcomp}
\usepackage{stfloats}
\usepackage{url}
\usepackage{verbatim}
\usepackage{graphicx}
\usepackage{cite}
\begin{document}

\title{Optimal Configuration of Reconfigurable Intelligent Surfaces with Arbitrary Discrete Phase Shifts}
\author{Seyedkhashayar Hashemi, Hai Jiang, and Masoud Ardakani\thanks{The authors are with the Department of Electrical and Computer Engineering, University of Alberta, Edmonton, AB T6G 1H9, Canada. (e-mail: \{seyedkha, hai1, ardakani\}@ualberta.ca)}}

\maketitle

\begin{abstract}
We address the reflection optimization problem for a reconfigurable intelligent surface (RIS), where the RIS elements feature a set of non-uniformly spaced discrete phase shifts. This is motivated by the actual behavior of practical RIS elements, where it is shown that a uniform phase shift assumption is not realistic. A problem is formulated to find the
optimal refection amplitudes and reflection phase shifts of the RIS elements such that the channel capacity of the target user is maximized. We first prove that in the optimal configuration, each RIS element is either turned off or operates at maximum amplitude. We then develop a method that finds the optimal reflection amplitudes and phases with complexity linear in the number of RIS elements. Some new and interesting insight into the reflection optimization problem is also provided.


\end{abstract}

\begin{IEEEkeywords}
Reconfigurable intelligent surfaces, discrete phase shifts, exhaustive search, arbitrary phase shifts, optimization.
\end{IEEEkeywords}

\section{Introduction}
\IEEEPARstart{R}{econfigurable}
intelligent surfaces (RIS) have gained lots of attention during the past few years. This new paradigm is one of the candidate technologies for future wireless communication systems\cite{ref1}. An RIS is a metasurface with a number of small passive elements, and is able to shape incoming radio signals by applying controllable phase shifts via the elements \cite{ref2}. Similar to conventional relays, RISs help us have control over the propagation environment but at a lower cost due to the passiveness of the RIS elements \cite{ref4}. In particular, these elements do not require any power amplifiers, making them less costly, more energy-efficient, and environmentally friendly\cite{ref3}.

Although RISs provide numerous benefits, there are still some design challenges that have to be dealt with to facilitate an efficient integration of RISs into wireless systems. Some of these challenges are reflection optimization (i.e., determining the optimal configuration of the RIS elements), channel estimation, and deployment\cite{ref5}. In this paper, we focus on reflection optimization. In specific, for the $n$th element of an RIS, denote $\theta_n = \beta_n e^{j\alpha_n}$ as the reflection coefficient, in which $\beta_n\in [0,~1]$ is the reflection amplitude and  $\alpha_n\in [0,~2\pi)$ is the reflection phase shift. Reflection optimization is to find the optimal $\theta_n$ for all RIS elements.

Assuming that an RIS consists of elements with continuous phase shifts, reflection optimization will not be a challenging task. The optimal configuration solution can be obtained by aligning all controllable paths (i.e, the paths aided by the RIS elements) with the direct path (the uncontrollable path) where all RIS elements are used at their maximal amplitude \cite{ref6, ref10}. However, it is more practical to assume discrete phase shifts for the RIS elements\cite{ref7}. This limitation is due to the hardware structure of RIS elements \cite{ref16}, making the continuous-phase-shift assumption unrealistic. With discrete phase shifts, simply aligning all controllable paths with the uncontrollable path may not be possible. Hence, the optimal solution is not readily available\cite{ref5}.


Several approaches can be used to find the optimal solutions in the discrete-phase-shift scenario. One method is to perform an exhaustive search over all possibilities of all RIS elements. The complexity of this method grows exponentially with the number of RIS elements. More specifically, assuming $K$ discrete phase shifts and $N$ RIS elements, the complexity of the exhaustive search is $O(K^N)$. Given that an RIS can have hundreds of elements, this method is extremely time-consuming\cite{ref8}. Discrete optimization methods like branch-and-bound (BB) can also be applied, but the worst-case time complexity would still be exponential with $N$. A number of sub-optimal solutions have also been studied. Closest point projection (CPP) is one of these methods. CPP is a heuristic approach that aligns each controllable path as much as possible with the uncontrollable path\cite{ref11}. In other words, this method quantizes the solution derived from the continuous-phase-shift optimization problem\cite{ref5}. The time complexity of this algorithm is $O(N)$. Iterative algorithms can also be used. Alternating optimization (AO) is an example of such algorithms. In this method, a single element is optimized at a time while the rest of the elements are set to a constant \cite{ref8}. The same procedure is applied to the rest of the elements until convergence. There exist some other sub-optimal solutions for the discrete-phase-shift optimization problem such as penalty-based methods \cite{ref12, ref13}, alternating direction method of multipliers (ADMM) \cite{ref14}, and more. Only recently, a method has been proposed which is able to determine the optimal configuration \cite{ref15}. The complexity of this solution is linear with the number of RIS elements.

All the aforementioned methods assume uniform phase shifts for the RIS elements, i.e., the discrete phase shifts are evenly spaced within the range $[0, 2\pi)$. However, in practical scenarios, phase shifts are not necessarily uniform \cite{ref16}. When it comes to hardware, there are various types of RIS based on their structure and the materials used. For example, ``varactor diode-based RIS"\cite{ref17, ref18, ref19}, ``pin diode-based RIS"\cite{ref20, ref21, ref22}, and ``liquid crystal-based RIS"\cite{ref23, ref24, ref25} are some of the experimented structures for RIS. Simulations and measurements in \cite{ref17, ref18, ref19,ref20, ref21, ref22,ref23, ref24, ref25} verify that the set of discrete phase shifts provided by existing RIS technologies is not necessarily uniform. Also, the phase shifts of elements depend on the working frequency. For example, consider an RIS element that has been designed to provide uniform phase shifts at 4 GHz. Now if the same element is used at 4.1 GHz, the phase shifts would not remain uniform \cite{ref28}.

The above observation has motivated us to investigate reflection optimization with arbitrary non-uniform discrete phase shifts. In this paper, we develop a method that can find the optimal reflection configuration of the RIS elements with linear complexity. The novelty and contributions of this paper are summarized as follows.
\begin{itemize}
\item Our investigated problem is more practical compared with existing research efforts, considering the fact that discrete phase shifts are not necessarily uniform in a realistic system. As a result, our problem is also more general than problems in existing research.

\item In previous studies with uniform phase shifts, the reflection amplitudes of the RIS elements are set to $\beta_n=1$ for all $n$. This setting is made because of the fact that maximizing the reflected signal strength is always beneficial if the optimal phase shift is adopted. However, we demonstrate that this is not the case for non-uniform phase shifts. In this paper, we prove that each RIS element should either take the maximal reflection signal strength or be simply turned off.

\item To maximize the channel capacity, we propose a method with linear complexity to find the optimal configuration of the RIS elements.


\item Our proposed method has a number of search steps. We develop a fast algorithm to further reduce the number of computations needed in these search steps. Our algorithm  also provides a fast sorting for a sorting requirement of our proposed method.



\item For the optimal overall channel coefficient (which is a complex number), intuitively one would think that the argument may take any value within $[0,2\pi)$. However, counter-intuitively, we show that there are many regions within $[0,2\pi)$, called empty regions,  in which the argument cannot be located. We also give expressions for these empty regions.


\end{itemize}

The remainder of this paper is structured as follows. Section~II presents the RIS system with arbitrary non-uniform discrete phase shifts and formulates the problem. In Section III, our proposed method is presented to optimally solve the problem with linear complexity. Section IV presents an interesting insight into the existence of the so-called empty regions. Simulation results are presented and discussed in Section~V. The conclusions of this work can be found in Section VI.

\section{System Model and Problem Formulation}
\subsection{System Model}

\begin{figure}[!t]
\centering
\includegraphics[width=3.4in]{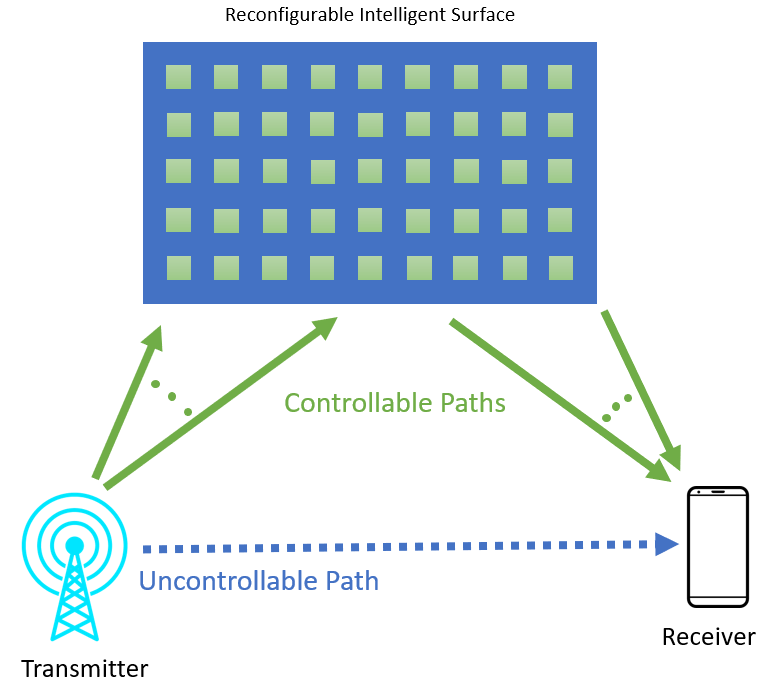}
\caption{The considered RIS-aided communication system.}
\label{system_fig}
\end{figure}

In this paper, we consider a system where a transmitter is communicating with a receiver, aided by an RIS between the transmitter and the receiver, as seen in Fig.~\ref{system_fig}. The RIS has $N$ elements. The reflection coefficient of the $n$th element\footnote{In the sequel, when we say the $i$th element, it means the $i$th RIS element.}, denoted as $\theta_n \in \mathbb{C}$, has the following pattern:
\begin{equation}
\label{deqn_ex1}
\begin{aligned}
&\theta_n = \beta_n e^{j\alpha_n}, n\in \{1,2,...,N\},\\
\end{aligned}
\end{equation}
in which the reflection amplitude $\beta_n\in [0,~1]$ and the reflection phase shift $\alpha_n\in [0,~2\pi)$ are independent. The channel coefficient of the direct path from the transmitter to the receiver is denoted as $h_d \in \mathbb{C}$. The channel coefficient of the path from the transmitter to the $n$th RIS element and then to the receiver is denoted as $g_n \in \mathbb{C}$. Here, $g_n$ is expressed as
\begin{equation}
\label{deqn_ex2}
g_n = h^\prime_n\theta_nh^{\prime\prime}_{n},
\end{equation}
where $h^\prime_n$ is the channel coefficient from the transmitter to the  $n$th RIS element, and $h^{\prime\prime}_{n}$ is the channel coefficient from the $n$th RIS element to the receiver. Since $h^\prime_n$ and $h^{\prime\prime}_{n}$ are usually used together in our subsequent analysis, we define $v_n$ as the concatenated channel coefficient, expressed as
\begin{equation}
\label{deqn_ex56}
v_n = h^\prime_nh^{\prime\prime}_{n}.
\end{equation}
Thus, $g_n$ can be expressed as
\begin{equation}
\label{gn_ex_new}
g_n=v_n \theta_n=v_n \beta_n e^{j\alpha_n}.
\end{equation}

Thus, between the transmitter and the receiver, the overall channel coefficient (including the direct path and all RIS paths) can be written as:
\begin{equation}
\label{deqn_ex55}
h = h_d + \sum_{n=1}^{N} g_n.
\end{equation}

The received signal at the receiver can be written as:
\begin{equation}
\label{deqn_ex2}
Y = hX + Z = (h_d + \sum_{n=1}^{N} g_n)X + Z,
\end{equation}
in which $X$ is the transmitted signal at the transmitter, and $Z \sim \mathcal{N_C}(0,\,N_0)\,$ is the complex Gaussian noise. Denoting $P$ as the transmit power at the transmitter and $B$ as the bandwidth, the channel capacity can be expressed as:
\begin{equation}
\label{deqn_ex3}
C = B \log_2(1+\text{SNR}) = B \log_2(1+\frac{P|h|^2}{BN_0}).
\end{equation}

\subsection{Problem Formulation}

Our goal is to achieve the maximum channel capacity. Thus, we need to find the optimal $\theta_n$ for each element, that results in the highest possible channel capacity. Accordingly, the following problem is formulated:
\begin{equation}
\label{deqn_ex77}
\begin{aligned}
\max_{\theta_1,\theta_2...,\theta_N} \quad & C\\
\textrm{s.t.} \quad &\beta_n \in [0,1],\\
  &\alpha_n \in [0,2\pi).\\
\end{aligned}
\end{equation}

According to \eqref{deqn_ex3}, in order to maximize the channel capacity, $|h|$ should be maximized.

Since the channel coefficients are all complex numbers, they can be represented by vectors in a complex plane. Thus, in the sequel, we also call a complex number, say complex number $a$, as vector $a$. For vector (complex number) $a$, we will use the following notation
\begin{equation}
\angle {{a}}  = arg({a}) \mod{2\pi}\\
\end{equation}
as the argument of complex number $a$, which is the counterclockwise angle from the positive real axis to vector ${a}$ in the complex plane.

Any two vectors in the complex plane make two angles as seen in Fig.~\ref{fig:fig2}, and the two angles add up to $2\pi$. We define the angle that is not larger than $\pi$ as {\it the angle between the two vectors}. Thus, for vectors (complex numbers) $a$ and $b$, the angle between them can be expressed as
\begin{equation}
ang({a} , {b}) = \min\{| \angle {a} - \angle {b} | , 2\pi - | \angle {a} - \angle {b} |\}.
\end{equation}

\begin{figure}[!t]
\centering
\includegraphics[width=2.5in]{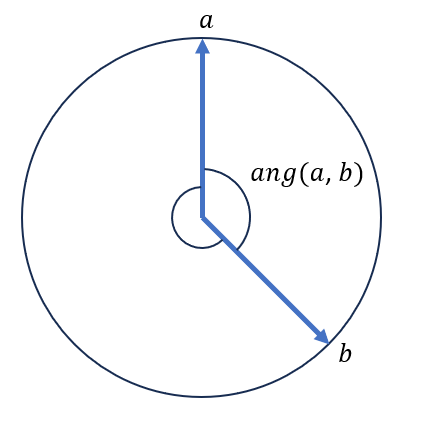}
\caption{The angle between two vectors.}
\label{fig:fig2}
\end{figure}

From (\ref{deqn_ex55}), vector $h$ is the summation of the following $N+1$ vectors: $h_d$ and $g_1,g_2,...,g_N$. From (\ref{gn_ex_new}), vector $g_n$ is obtained by rotating vector $v_n$ counterclockwise in the complex plane by angle $\alpha_n$ with an amplitude amplifier (amplifying factor being $\beta_n$). Thus, we should configure $\alpha_n$ and $\beta_n$ so as to maximize $|h|$. If each RIS element can take an arbitrary {\it continuous} reflection phase shift, then it is optimal to 1) rotate vector $v_n$ such that the resulted vector $g_n$ overlaps\footnote{When we say two vectors overlap, it means that the counterclockwise angles from the positive real axis to the two vectors are equal.} with $h_d$ and 2) apply an amplifying factor $1$. In other words, we should have $\alpha_n=\angle h_d - \angle v_n \text{~~mod~} 2\pi$ and $\beta_n=1$, and accordingly, the optimal configuration of the $n$th element is expressed as
\begin{equation}
\label{deqn_ex28}
\theta^*_n = e^{j[(\angle h_d - \angle v_n) \mod 2\pi]}.
\end{equation}

However, it may not be practical for an RIS element to have an arbitrary continuous reflection phase shift. Thus, in the literature, RIS elements with discrete phase shifts are assumed. A discrete set of uniform reflection phase shifts is assumed in \cite{ref7, ref8, ref15}, in which the reflection phase shifts are evenly spaced within range $[0,~2\pi)$. For example, a discrete set of $M$ uniform reflection phase shifts could be $\{0, \frac{2\pi}{M}, \frac{2\times 2\pi}{M}, \frac{3\times 2\pi}{M},...,\frac{(M-1)\times 2\pi}{M}\}$.

On the other hand, in a real system, it is not easy to guarantee that the reflection phase shifts of an RIS element are uniform.  For example, based on the transmission line model presented in \cite{ref27}, Fig.~\ref{fig:fig3} shows the real implementation of an RIS element that was designed to provide two evenly-spaced reflection phase shifts: $\alpha_1 = 90$\textdegree, $\alpha_2 = 270$. Here {\it phase gap} is defined as the difference of the two reflection phase shifts in a real system. From Fig.~\ref{fig:fig3}, it is seen that the phase gap is 180\textdegree ~ when the operating frequency is at 4 GHz but shrinks significantly as the operating frequency changes to 3 GHz or 5 GHz. Thus, it is more practical to consider reflection phase shifts that may not be uniform in a real system. Therefore, we consider that each RIS element can select from $K$ reflection phase shifts, denoted as $\phi_1,\phi_2, ..., \phi_{K}$, and in general the $K$ reflection phase shifts are not uniform. Accordingly, we reformulate our optimization problem in \eqref{deqn_ex77} as the following problem.
\begin{equation}
\label{deqn_ex69}
\begin{aligned}
\max_{\theta_1,\theta_2,...,\theta_N} \quad & |h|\\
\textrm{s.t.} \quad &\beta_n \in [0,1],\\
  &\alpha_n \in \{\phi_1,\phi_2, ..., \phi_{K} \}, \\
  &0\leq \phi_1 <\phi_2< ... < \phi_{K} < 2\pi.
\end{aligned}
\end{equation}

\begin{figure}[!t]
\centering
\includegraphics[width=3.4in]{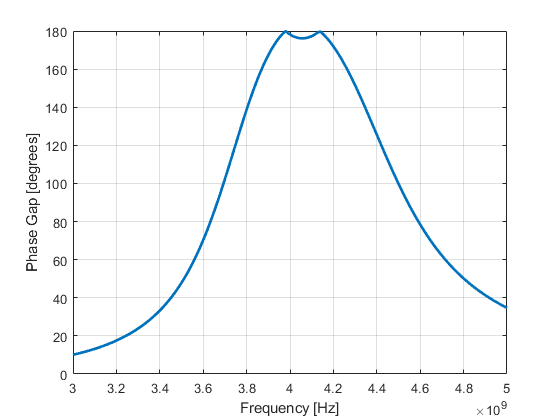}
\caption{Phase gap versus operating frequency in a real system.}
\label{fig:fig3}
\end{figure}

\section{Proposed Method to Get Optimal Solution}
In this section, we will propose a method to get an optimal solution for the problem in (\ref{deqn_ex69}).


\subsection{Problem Transformation}

As we can see in \eqref{deqn_ex69}, reflection amplitudes $\beta_n$'s are chosen from a continuous range $[0,~1]$. For the existing research efforts that consider uniform discrete reflection phase shifts, reflection amplitudes are always one.
Contrary to these existing results, in the following lemma we will prove that for each element in our system, its optimal reflection amplitude is either $0$ or $1$. 
\newtheorem{case}{Case}

\newtheorem{lem}{Lemma}
\newtheorem{theo}{Theorem}

\begin{lem}
In the optimal configuration, the reflection amplitude of each element is either 0 or 1.

\end{lem}
\begin{proof}
Consider the $i$th element ($i\in\{1,2,...,N\}$). Define
\begin{equation}
\Bar{g}_{i} = h_d + \sum_{n=1,n\neq i}^{N} g_n.
\label{eq13}
\end{equation}
Thus, from \eqref{deqn_ex55} and \eqref{eq13} we have $h=g_i + \Bar{g}_{i}$

We consider two cases: $|\Bar{g}_{i}|=0$ and $|\Bar{g}_{i}|>0$.

\begin{case}
   $|\Bar{g}_{i}| = 0$
   \end{case}

   In this case, $h = g_{i} = v_i \beta_i e^{j\alpha_i}$. In order to maximize $|h|$, the reflection amplitude ($\beta_{i}$) of the $i$th element should be set to the maximum value which is 1.
\begin{case}
   $|\Bar{g}_{i}| > 0$
\end{case}

From $h=g_i + \Bar{g}_{i}$, we have
\begin{equation}
\label{deqn_ex19}
\begin{aligned}
&|h|^2 = |g_{i} + \Bar{g}_{i}|^2\\
&= |g_{i}|^2 + |\Bar{g}_{i}|^2 \\
&+ 2|g_{i}|\times |\Bar{g}_{i}|\cos(ang(g_{i},\Bar{g}_{i}))\\
&= \beta^2_{i}|v_ie^{j\alpha_{i}}|^2 + |\Bar{g}_{i}|^2 \\
&+ 2\beta_{i}|v_ie^{j\alpha_{i}}|\times |\Bar{g}_{i}|\cos(ang(g_{i},\Bar{g}_{i}) ).
\end{aligned}
\end{equation}

   Now consider the following two possible scenarios: $ang(g_{i},\Bar{g}_{i}) \leq \frac{\pi}{2} $ and $ang(g_{i},\Bar{g}_{i}) > \frac{\pi}{2} $.

When $ang(g_{i},\Bar{g}_{i}) \leq \frac{\pi}{2} $, it can be seen from  \eqref{deqn_ex19} that $|h|^2$ is the summation of three non-negative expressions. Thus, in order to maximize $|h|^2$, $\beta_{i}$ should be set to its maximum value which is 1.

When $ang(g_{i},\Bar{g}_{i}) > \frac{\pi}{2} $, from \eqref{deqn_ex19} it can be seen that $|h|^2$ is a convex function of $\beta_{i}$. For this convex function to attain its maximum value over range $\beta_{i}\in [0,1]$, $\beta_{i}$ should take either value 0 or value 1.

This completes the proof.


\end{proof}
As we can see in Lemma 1, for some elements $\beta_{i}=0$ could be the optimal solution. Whether or not we have elements with $\beta_{i}=0$ depends on the way the phase shifts are distributed. For example, if the phase shifts are distributed uniformly, $\beta_{i}=1$ would be the optimal solution for all elements. This is why existing work, where phases are distributed uniformly, always considers $\beta_i=1$. However, with non-uniform phase shifts, for some elements $\beta_i=0$ will be optimal, to be shown in Section \ref{sec:optimal_with_angleknown} and \ref{sec:optimalgnfromknowangle}.

According to Lemma 1, we can transform our formulated problem in \eqref{deqn_ex69} to the following problem:
\begin{equation}
\label{deqn_ex14}
\begin{aligned}
\max_{\theta_1,\theta_2,...,\theta_N} \quad & |h|\\
\textrm{s.t.} \quad &\beta_n \in \{0,1\},\\
  &\alpha_n \in \{\phi_1,\phi_2, ..., \phi_{K} \}, \\
 &0\leq \phi_1 <\phi_2< ... < \phi_{K} < 2\pi.
\end{aligned}
\end{equation}
Note that when the reflection amplitude of an element is set to 0, the element's reflection phase shift will not be important anymore. We will refer to this state of the element as the ``off" state. Thus, for the reflection coefficient $\theta_n$ of the $n$th element, its optimal value is taken from the set of $K+1$ values: $0,e^{j\phi_1}, e^{j\phi_2},..., e^{j\phi_{K}}$. Accordingly, the resulted $g_n$ is from a set of $K+1$ values: $0, F_{n,1}, F_{n,2},...,F_{n,K}$, in which
\begin{equation}\label{e:vn_to_Fni}
F_{n,i} = v_n e^{j\phi_i}, i=1,2,...,K.
\end{equation}
It can be seen that $F_{n,i}$ is actually the resulted $g_n$ when the $n$th element applies reflection coefficient $e^{j\phi_i}$.

Thus, the problem in (\ref{deqn_ex14}) is  equivalent to the following problem:
\begin{equation}
\label{deqn_ex49}
\begin{aligned}
\max_{g_1,g_2,...,g_N} \quad & |h|\\
\textrm{s.t.} \quad &g_n \in \{0,F_{n,1},F_{n,2},...,F_{n,K}\} .
\end{aligned}
\end{equation}

For the problem in (\ref{deqn_ex49}), denote the optimal $g_n$ as $g_n^*$, and denote the optimal resulted $h$ as $h^*$. To solve the problem, in Section \ref{sec:optimal_with_angleknown} we will first find $g_n^*$ by assuming $\angle h^*$ (i.e., the counterclockwise angle from the positive real axis to vector $h^*$) is known. Then in Section \ref{sec:optimalgnfromknowangle} the overall optimal configuration of the considered system is found by comparing the achieved $|h^*|$ values associated with all possible cases of $\angle h^*$.



\subsection{Getting $g_n^*$ when $\angle h^*$ is known}\label{sec:optimal_with_angleknown}

Here we assume that we know $\angle h^*$ but do not know $|h^*|$. Next, we first present two lemmas (Lemma 2 and Lemma 3) showing some properties of $g_n^*$, and give a theorem (Theorem 1) for finding optimal $g_n^*$.

\begin{lem}
Consider the $n$th element. Among $F_{n,1}, F_{n,2},...,F_{n,K}$, if there exists $F_{n,i}$ such that the angle between ${h}^*$ and $F_{n,i}$ is less than $\frac{\pi}{2}$, then the element should not be turned off, i.e., $g_n^*\neq 0$, in the optimal configuration.
\end{lem}
\begin{proof}
We use proof by contradiction. For the $n$th element, suppose the angle between ${h}^*$ and $F_{n,i}$ is less than $\frac{\pi}{2}$. Assume the optimal configuration of the $n$th element is $g_n^*=0$ (i.e., the element is turned off). Recall that $h^*$ is optimal $h$.

In the optimal configuration of the system, if we turn on the $n$th element and make $g_n=F_{n,i}$ (i.e., the $n$th element applies reflection coefficient $e^{j\phi_i}$), then the resulted $h$ would be given as $h^\dag=h^* + F_{n,i}$, as shown in Fig.~\ref{lem2pr}. Since the angle between ${h}^*$ and $F_{n,i}$ is less than $\frac{\pi}{2}$, we know that $|h^\dag| > |h^*|$, which contradicts the fact that $h^*$ is the optimal $h$.

\end{proof}
\begin{figure}[!t]
\centering
\includegraphics[width=1.7in]{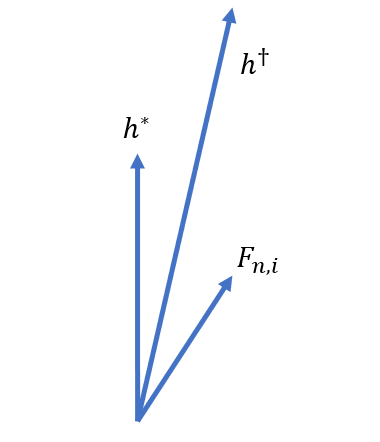}
\caption{$h^*$, $F_{n,i}$, and $h^\dag$ in proof of Lemma 2.}
\label{lem2pr}
\end{figure}

\begin{lem}
Consider the $n$th element. If the angle between ${h}^{*}$ and every $F_{n,i}$ ($i=1,2,...,K$) is not smaller than $\frac{\pi}{2}$, then the $n$th element should be turned off, i.e., $g_n^*=0$, in the optimal configuration.
\end{lem}
\begin{proof}
We use proof by contradiction. Assume the $n$th element is not turned off in the optimal configuration, and suppose $g_n^*=F_{n,l}$, $l\in\{1,2,...,K\}$. So the angle between ${h}^{*}$ and $F_{n,l}$  is not smaller than $\frac{\pi}{2}$.  Recall that $h^*$ is optimal $h$.

In the optimal configuration of the system, if we turn off the $n$th element, then the resulted $h$ is expressed as $h^\dag = h^* - F_{n,l}$, as shown in Fig.~\ref{lem3pr}. Since the angle between ${h}^{*}$ and $F_{n,l}$ is greater than or equal to $\frac{\pi}{2}$, we have $|h^\dag| > |h^*|$, which contradicts  the fact that $h^*$ is the optimal $h$.

\end{proof}
\begin{figure}[!t]
\centering
\includegraphics[width=1.7in]{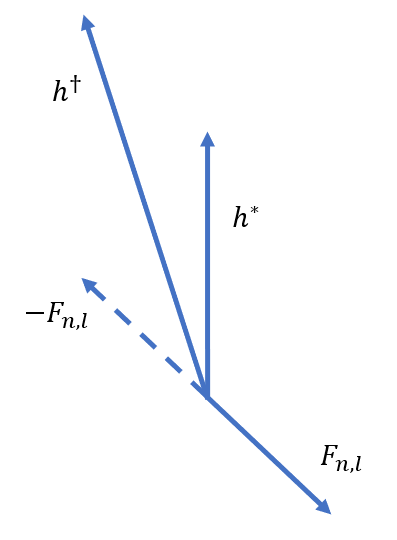}
\caption{$h^*$, $F_{n,l}$, and $h^\dag$ in proof of Lemma 3.}
\label{lem3pr}
\end{figure}

\begin{theo}
For the $n$th element, among vectors $F_{n,1}, F_{n,2},...,F_{n,K}$, denote $F_{n,i}$ as the vector that has the smallest angle with ${h}^*$. We have the following results. (i) If the angle between $F_{n,i}$ and ${h}^*$ is less than $\pi/2$, then $g_n^*=F_{n,i}$, and $F_{n,i}$ is the only vector (among $F_{n,1}, F_{n,2},...,F_{n,K}$) that has the smallest angle with vector ${h}^*$. \\
(ii) The angle between $F_{n,i}$ and ${h}^*$ is not equal to $\pi/2$. \\
(iii) If the angle between $F_{n,i}$ and ${h}^*$ is more than $\pi/2$, then $g_n^*=0$ (i.e., the $n$th element should be turned off).
\end{theo}
\begin{proof}
$ $\newline
{\bf Proof for  Part (i)}.

For Part (i), the angle between $F_{n,i}$ and ${h}^*$ is less than $\pi/2$. From Lemma 2 we know that the $n$th element should not be turned off.

Firstly, we consider that $F_{n,i}$ is the only vector (among $F_{n,1}, F_{n,2},...,F_{n,K}$) that has the smallest angle with vector ${h}^*$ (this statement will be proved later). We will use proof by contradiction to prove $g_n^*=F_{n,i}$.

Assume $g_n^*\neq F_{n,i}$. Suppose $g_n^*=F_{n,l}$ ($l\in\{1,2,...,K\}, l\neq i$). So the angle between $F_{n,i}$ and ${h}^*$ is smaller than the angle between $F_{n,l}$ and ${h}^*$. We consider two scenarios: vectors $F_{n,i}$ and $F_{n,l}$ are on the same side of vector ${h}^*$ (Scenario I), or on different sides of vector ${h}^*$ (Scenario II).

Consider Scenario I when vectors $F_{n,i}$ and $F_{n,l}$ are on the same side of vector ${h}^*$, as shown in Fig.~\ref{scenario1}. Thus, vector $F_{n,i}$ is closer to vector ${h}^*$ than $F_{n,l}$ to vector ${h}^*$. Recall that we have $g_n^*= F_{n,l}$ in the optimal configuration, and $h^*$ is the optimal $h$. In the optimal configuration, if we change $g_n$ from $F_{n,l}$ to $F_{n,i}$, the resulted $h$ is expressed as $h^\dag = h^*- F_{n,l} + F_{n,i} $. We look at vector $h^*- F_{n,l}$, which is the summation of vector $h^*$ and vector $-F_{n,l}$. As shown in Figure \ref{scenario1}, vector $h^*- F_{n,l}$ is located in between vector $h^*$ and vector $-F_{n,l}$. It can be seen that vector $(h^*- F_{n,l})$ is closer to vector $F_{n,i}$ than it is to vector $F_{n,l}$. Since vector $F_{n,i}$ and vector $F_{n,i}$ have the same amplitude, we can conclude that the summation of vector $(h^*- F_{n,l})$ and vector $F_{n,i}$ (the summation expressed as $h^*- F_{n,l} + F_{n,i}=h^\dag$) has a large amplitude than the amplitude of summation of vector $(h^*- F_{n,l})$ and vector $F_{n,l}$ (the summation expressed as $h^*- F_{n,l} + F_{n,l} = h^*$). In other words, we have $|h^\dag|> |h^*|  $, which contradicts the fact that $h^*$ is the optimal $h$.

\begin{figure}[!t]
\centering
\includegraphics[width=1.9in]{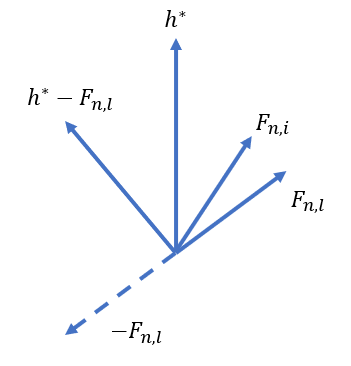}
\caption{Scenario I in proof of Part (i) of Theorem 1.}
\label{scenario1}
\end{figure}
Now we consider Scenario II when vectors $F_{n,i}$ and $F_{n,l}$ are on different sides of vector ${h}^*$. Recall that  $F_{n,i}$ is the only vector (among $F_{n,1}, F_{n,2},...,F_{n,K}$) that has the smallest angle with vector ${h}^*$. Thus, vector $F_{n,i}$ is closer to vector ${h}^*$ than $F_{n,l}$ to vector ${h}^*$. In the optimal configuration in which $g_n^*=F_{n,l}$, if we change $g_n$ from $F_{n,l}$ to $F_{n,i}$, the resulted $h$ is expressed as $h^\dag = h^*- F_{n,l} + F_{n,i} $. We have two possible situations for the location of vector $(h^*- F_{n,l})$.
\begin{itemize}
\item In Situation 1, vector $(h^*- F_{n,l})$ is located between vector $-F_{n,l}$ and vector $F_{n,i}$, as shown in the left-hand side of Fig.~\ref{scenario2}. In this situation, apparently $(h^*- F_{n,l})$ is closer to vector $F_{n,i}$ than it is to vector $F_{n,l}$.

\item In Situation 2, vector $(h^*- F_{n,l})$ is located between vector $F_{n,i}$ and vector $h^*$, as shown in the right-hand side of Fig.~\ref{scenario2}. Since vector $F_{n,i}$ is closer to vector ${h}^*$ than $F_{n,l}$ to vector ${h}^*$, it can be seen that $(h^*- F_{n,l})$ is closer to vector $F_{n,i}$ than it is to vector $F_{n,l}$.
\end{itemize}
In both situations, $(h^*- F_{n,l})$ is always closer to vector $F_{n,i}$ than it is to vector $F_{n,l}$. Similar to Scenario I, we have $|h^\dag|> |h^*|$, which contradicts the fact that $h^*$ is the optimal $h$.
\begin{figure}[!t]
\centering
\includegraphics[width=3.6in]{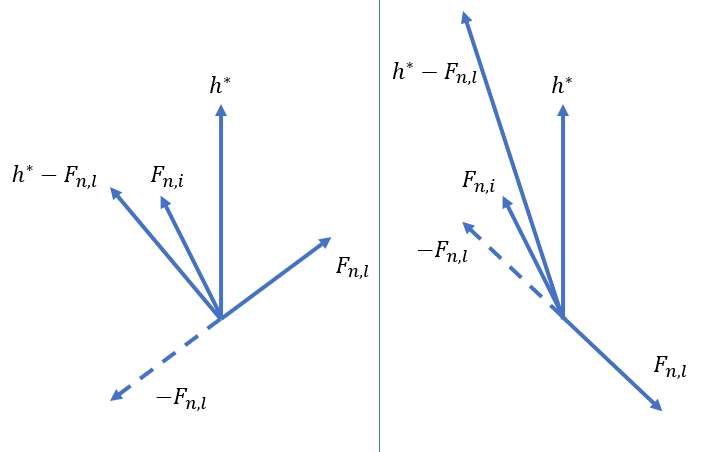}
\caption{Scenario II in proof of Part (i) of Theorem 1.}
\label{scenario2}
\end{figure}

Summarizing Scenario I and Scenario II, we have $g_n^*=F_{n,i}$ if $F_{n,i}$ is the only vector (among $F_{n,1}, F_{n,2},...,F_{n,K}$) that has the smallest angle with vector ${h}^*$.

Now we prove that if $F_{n,i}$ has the smallest angle with ${h}^*$ and the angle is less than $\pi/2$, then $F_{n,i}$ is the only vector (among $F_{n,1}, F_{n,2},...,F_{n,K}$) that has the smallest angle with vector ${h}^*$. We use proof by contradiction. Assume there is another $F_{n,l}$ ($l\in \{1,2,...,K\}, l\neq i$), and vectors $F_{n,i}$ and $F_{n,l}$ have the same angle to vector $h^*$. So vectors $F_{n,i}$ and $F_{n,l}$ are on different sides of vector ${h}^*$, as shown in Fig.~\ref{lemma4pt1}. From the above proof of Part (i), it can be seen that among $F_{n,1}, F_{n,2},...,F_{n,K}$, if $g_n$ takes either $F_{n,i}$ or $F_{n,l}$, the achieved $|h|$ is higher than that when $g_n$ takes any of the other $K-2$ values. Thus, $g_n^*$ should be either $F_{n,i}$ or $F_{n,l}$.

\begin{figure}[!t]
\centering
\includegraphics[width=2in]{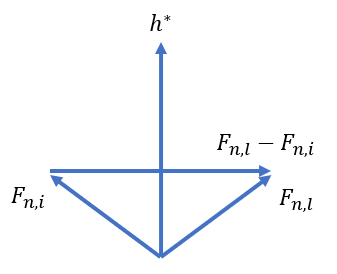}
\caption{Demonstration when vectors $F_{n,i}$ and $F_{n,l}$ have the same smallest angle to vector $h^*$.}
\label{lemma4pt1}
\end{figure}

\begin{itemize}
\item When $g_n^*$ is $F_{n,i}$, it means that $h^*$ is achieved when $g_n=F_{n,i}$. In this optimal configuration, if $g_n$ changes from $F_{n,i}$ to $F_{n,l}$, then the achieved $h$ is expressed as $h^\dag = h^* - F_{n,i} + F_{n,l} = h^*+ (F_{n,l} - F_{n,i})$. From Fig.~\ref{lemma4pt1} it is seen that the angle between vector $h^*$ and vector $(F_{n,l} - F_{n,i})$ is $\pi/2$. Thus, vector $h^*+ (F_{n,l} - F_{n,i})$ has a larger amplitude than that of $h^*$. In other words, $|h^\dag| > |h^*|$, which contradicts the fact that $h^*$ is the optimal $h$.

\item When $g_n^*$ is $F_{n,l}$, similarly it also leads to a contradiction.
\end{itemize}
Since a contradiction is always the result, we can conclude that $F_{n,i}$ is the only vector (among $F_{n,1}, F_{n,2},...,F_{n,K}$) that has the smallest angle with vector ${h}^*$.

{\bf Proof of Part (ii).}

We use proof by contradiction. Assume the angle between $F_{n,i}$ and ${h}^*$ is equal to $\pi/2$. Based on Lemma 3, it can be seen that the $n$th element should be turned off in the optimal configuration of the system in which $h^*$ is achieved. In the optimal configuration of the system, if we turn on the $n$th element and make $g_n=F_{n,i}$, then the achieved $h$ is expressed as $h^\dag = h^* + F_{n,i}$. Since the angle between $F_{n,i}$ and ${h}^*$  is equal to $\pi/2$, we have $|h^\dag| > |h^*|$, which contradicts the fact that $h^*$ is the optimal $h$.

{\bf Proof of Part (iii).}

If the angle between $F_{n,i}$ and ${h}^*$ is larger than $\pi/2$, then from Lemma 3 we have $g_n^*=0$.

\end{proof}

\subsection{Getting $g_n^*$ for All Possible $\angle h^*$ with Linear Complexity}\label{sec:optimalgnfromknowangle}
In the preceding subsection, we get $g_n^*$ when $\angle h^*$ is known. However, $\angle h^*$ is unknown in advance. Thus, to get $g_n^*$ of the system, theoretically we should exhaustively search all possible values of $\angle h^* \in [0, 2\pi)$ and take the $g_n^*$ that maximizes $|h|$. Although there are an infinite number of $\angle h^*$ values in $[0, 2\pi)$, next we show that we only need to search a finite set of possibilities.

In the complex plane, consider a circle with a center at the origin and a radius of 1. This is the circle that we refer to when we say ``the circle" in the sequel. Consider the $n$th element as the target element. Now we place $K$ vectors: $F_{n,1}, F_{n,2},...,F_{n,K}$ in the complex plane, as shown in Fig.~\ref{FNIK}. The $K$ vectors partition the circle into $K$ regions: the region from $F_{n,1}$ to $F_{n,2}$, the region from $F_{n,2}$ to $F_{n,3}$, ...,  the region from $F_{n,K-1}$ to $F_{n,K}$, and the region from $F_{n,K}$ to $F_{n,1}$. Among all the regions, at most one region is more than or equal to half of the circle, i.e., the angle of the region\footnote{Here a region is a sector. So the angle of the region is the angle of the sector.} is more than or equal to $\pi$. Without loss of generality, we consider that one region has an angle larger than $\pi$.\footnote{The case when all regions have their angles less than $\pi$ and the case when one region has its angle equal to $\pi$ can be treated similarly.}
\begin{figure}[!t]
\centering
\includegraphics[width=2.4in]{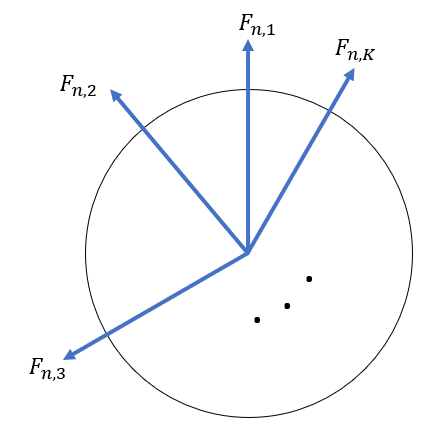}
\caption{Demonstration of $F_{n,1}, F_{n,2},...,F_{n,K}$.}
\label{FNIK}
\end{figure}

Recall that we should exhaustively search all possible $\angle h^*$ values within $[0, 2\pi)$. When $\angle h^*$ varies from $0$ to $2\pi$, it means that ${h}^*$ rotates in the complex plan counterclockwise from the positive real axis until it is back to the positive real axis. Consider that ${h}^*$ rotates counterclockwise within the region from $F_{n,i}$ to $F_{n,i+1}$. We have two scenarios as follows.
\begin{itemize}
\item If the counterclockwise angle from $F_{n,i}$ to $F_{n,i+1}$ is less than $\pi$, as shown in Fig.~\ref{lt90}, then the angle between ${h}^*$ and at least one of vectors $F_{n,i}$ and $F_{n,i+1}$ should be less than $\pi/2$. Thus, from Theorem 1, in the optimal configuration, $g_n^*$ is either $F_{n,i}$ or $F_{n,i+1}$, whichever is closer to ${h}^*$. Thus, in the region from $F_{n,i}$ to $F_{n,i+1}$, we can have the middle vector (i.e., the vector in the middle of $F_{n,i}$ and $F_{n,i+1}$) with amplitude being $1$, denoted as $S_{n,i} = e^{j(\angle F_{n,i} + ang(F_{n,i}, F_{n,i+1})/2)}$, as shown in Fig.~\ref{lt90} (the empty region in this figure will be discussed in Section \ref{sec:empty_region}). Vector $S_{n,i}$ is called the {\it separation line} in the region from $F_{n,i}$ to $F_{n,i+1}$. Note that this separation line is associated with the $n$th element, and thus, is called the $n$th element's separation line. From Theorem 1, we know that it is impossible for $h^*$ to overlap with the separation line $S_{n,i}$. Thus, when ${h}^*$ rotates from $F_{n,i}$ to the separation line, we have $g_n^* = F_{n,i}$; when $h^*$ rotates from the separation line to $F_{n,i+1}$, we have $g_n^* = F_{n,i+1}$. Accordingly, for presentation simplicity in the sequel, we call $F_{n,i}$ as the {\bf starting vector} of the separation line $S_{n,i}$, and call $F_{n,i+1}$ as the {\bf ending vector} of the separation line $S_{n,i}$. Here the starting vector of a separation line means the optimal configuration of the target element when ${h}^*$ rotates prior to the separation line, and the ending vector of a separation line means the optimal configuration of the target element when ${h}^*$ rotates past the separation line.

\begin{figure}[!t]
\centering
\includegraphics[width=2.4in]{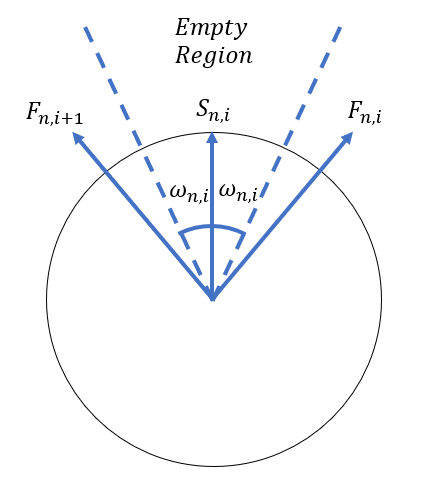}
\caption{Scenario when the counterclockwise angle from $F_{n,i}$ to $F_{n,i+1}$ is less than $\pi$.}
\label{lt90}
\end{figure}

\item If the counterclockwise angle from $F_{n,i}$ to $F_{n,i+1}$ is more than $\pi$, as shown in Fig.~\ref{mt90}, then in the region from $F_{n,i}$ to $F_{n,i+1}$, the $n$th element has two separation lines with amplitude being $1$: the first separation line $S_{n,i,1} = e^{j(\angle F_{n,i} + \pi/2)}$ has a $\pi/2$ angle to vector $F_{n,i}$, and the second separation line $S_{n,i,2} = e^{j(\angle F_{n,i+1} - \pi/2)}$ has a $\pi/2$ angle to vector $F_{n,i+1}$. From Theorem 1 we know that it is impossible for ${h}^*$ to overlap with any one of the two separation lines. From Theorem 1, we also have the following results. When ${h}^*$ rotates from $F_{n,i}$ to the first separation line, we have $g_n^* = F_{n,i}$; when ${h}^*$ rotates from the first separation line to the second separation line, we have $g_n^*=0$; when ${h}^*$ rotates from the second separation line to $F_{n,i+1}$, we have $g_n^*=F_{n,i+1}$. Thus, the starting vector and ending vector of the first separation line $S_{n,i,1}$ are $F_{n,i}$ and $0$, respectively, and the starting vector and ending vector of the second separation line $S_{n,i,2}$ are $0$ and $F_{n,i+1}$, respectively,

\end{itemize}
\begin{figure}[!t]
\centering
\includegraphics[width=2.4in]{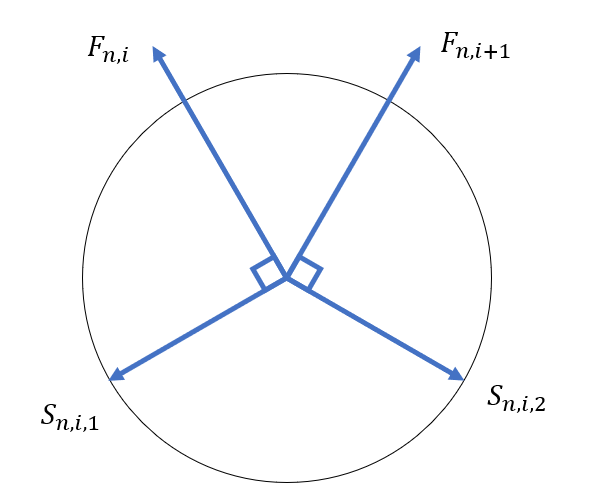}
\caption{Scenario when the counterclockwise angle from $F_{n,i}$ to $F_{n,i+1}$ is more than $\pi$.}
\label{mt90}
\end{figure}

Overall, for the $n$th element, there are $K+1$ separation lines in total. These separation lines partition the circle into $K+1$ sectors. When ${h}^*$ rotates in a particular sector, the optimal configuration of the $n$th element remains unchanged and is the starting vector of the separation line that is immediately after the sector in the counterclockwise direction. When ${h}^*$ rotates beyond the separation line, then the optimal configuration of the $n$th element changes to the ending vector of the separation line (which is also the starting vector of the $n$th element's next separation line).

For the considered system with $N$ elements, there are $N(K+1)$ separation lines in total. These separation lines partition the circle into $N(K+1)$ sectors. When ${h}^*$ is within a sector, the optimal configurations of the $N$ elements do not change. Thus, for ${h}^*$ in each sector, we can get the optimal configuration of all the elements and calculate the optimal $h$. Then among all the $N(K+1)$ calculated optimal $h$ values in the $N(K+1)$ sectors, we pick up the one with the largest amplitude, and the corresponding configuration of the $N$ elements is the optimal configuration of the considered system.

Overall, our proposed method needs to search $N(K+1)$ sectors. So the complexity is $O(N(K+1))$, which is linear with $N$.

\subsection{A Fast Algorithm for the Search Steps}

In our proposed method, we have $N(K+1)$ search steps for the $N(K+1)$ sectors, and in each search step, $N$ vector additions are needed to calculate ${h}^*$ for the corresponding sector. So in all $N(K+1)$ search steps, we need $N^2(K+1)$ vector additions in total. Further, to figure out where each sector starts and ends, we need to know the order of the  $N(K+1)$  separation lines on the circle in the complex plane, i.e., we need to sort the  $N(K+1)$  separation lines in ascending order of their arguments. Generally, to sort the $N(K+1)$ arguments of the separation lines, traditional sorting algorithms may need $N(K+1) \log [N(K+1)]$ complexity.

Next, we will develop a fast algorithm. Our fast algorithm needs much fewer vector additions to calculate the $N(K+1)$ values of ${h}^*$ in the $N(K+1)$ sectors and can sort the $N(K+1)$ arguments of the separation lines in a faster way.

Our fast algorithm is motivated by the following observation. There are $N(K+1)$  separation lines, which partition the circle into $N(K+1)$ sectors. Consider that ${h}^*$ rotates counterclockwise from one sector (say Sector 1) to another sector (say Sector 2), i.e., ${h}^*$ is passing a separation line. Denote the starting vector and ending vector of the separation line as $A$ and $B$, respectively. When ${h}^*$ is in Sector 1, denote the corresponding optimal $h$ as $h^*_1$. Then when ${h}^*$ is in Sector 2, the optimal $h$ is given as $h^*_2 = h^*_1 - A + B$. Thus, from the optimal $h$ value in Sector 1, we can get the optimal $h$ value in Sector 2 by using only two vector additions\footnote{From $h^*_1$ to $h^*_2$, we need one vector subtraction and one vector addition. Since vector subtraction and vector addition have the same computation burden, here one vector subtraction and one vector addition are counted as two vector additions.}.

Based on the above observation, our fast algorithm works as follows. Among the $N(K+1)$ sectors, we let $\angle {h}^*$ start within any specific sector and calculate its optimal $h$, which needs $N$ vector additions. Then we rotate ${h}^*$ counterclockwise until it returns to its original position. When ${h}^*$ passes a separation line, we need two vector additions to get the optimal $h$ in the next sector. Overall, to get optimal $h$ in all the $N(K+1)$ sectors, the number of vector additions that we need is $N + 2 \times N(K+1)=N(2K+3)$.

Next we show how our fast algorithm will sort the $N(K+1)$ arguments of the separation lines.

Recall that $v_n$ is the concatenated channel coefficient related to the $n$th element. Assume the $N$ elements are indexed such that $\angle v_1 < \angle v_2 <...<\angle v_N$. Consider the $n$th element. From (\ref{e:vn_to_Fni}), we have
\begin{equation}
\angle F_{n,i} = \angle v_n + \phi _ i { ~\mod~} 2\pi.
\end{equation}
Based on Section \ref{sec:optimalgnfromknowangle}, we have the following results:
\begin{itemize}

\item If the counterclockwise angle from $F_{n,i}$ to $F_{n,i+1}$ is less than $\pi$, the $n$th element has one separation line between $F_{n,i}$ and $F_{n,i+1}$, and the argument of the separation line is expressed as
\begin{equation}\label{e:Fargumenteq1}
\begin{aligned}
&(\angle v_n + {\frac{\phi_i + \phi_{i+1}}{2}}) ~\mod~2\pi,\quad i\in\{1,2,...,K-1\};\\
&(\angle v_n + {\frac{\phi_K + \phi_{1} + 2\pi}{2}}) ~\mod~2\pi,\quad i=K.
\end{aligned}
\end{equation}
\item If the counterclockwise angle from $F_{n,i}$ to $F_{n,i+1}$ is more than $\pi$, the $n$th element has two separation lines between $F_{n,i}$ and $F_{n,i+1}$, and the arguments of the separation lines are expressed as
\begin{equation}\label{e:Fargumenteq2}
(\angle v_n + {\phi_i + \pi/2}) ~\mod~2\pi
\end{equation}
and
\begin{equation}\label{e:Fargumenteq3}
(\angle v_n + {\phi_{i+1} - \pi/2}) ~\mod~2\pi,
\end{equation}
respectively.

\end{itemize}

Accordingly, for the $n$th element, from $F_{n,1}$ to $F_{n,2}$ and continuing to $F_{n,K}$ and back to $F_{n,1}$, we have $K+1$ separation lines. Denote the arguments of the $K+1$ separation lines as $A_{n,1}, A_{n,2},...,A_{n,K+1}$.

Next we show how our fast algorithm sorts $A_{n,k}$'s ($n\in \{1,2,...,N\}, k\in \{1,2,...,K+1\}$) in a faster way.

All the $A_{n,k}$'s ($n\in \{1,2,...,N\}, k\in \{1,2,...,K+1\}$) form a two-dimensional matrix as follows:
\begin{equation}\label{e:matrix_exp}
\mathcal{A}=
\begin{bmatrix}
A_{1,1} & A_{1,2} & \ldots
& A_{1,K+1} \\
A_{2,1} & A_{2,2} & \ldots
& A_{2,K+1} \\
\vdots & \vdots & \ddots
& \vdots \\
A_{N,1} & A_{N,2} & \ldots
& A_{N,K+1}
\end{bmatrix}.
\end{equation}

In a real system, usually $K$ is much smaller than $N$. So in matrix $\mathcal{A}$, the number of rows is more than the number of columns. We first try to sort each column of $\mathcal{A}$. For the $l$th column, if we assign its $N$ elements as arguments of $N$ unit-amplitude vectors, and place the $N$ unit-amplitude vectors in a complex plane, we can see that the $N$ unit-amplitude vectors are in counterclockwise order starting from the first unit-amplitude vector. So in the $l$th column of the matrix, all elements but one are smaller than their subsequent elements (here we consider $A_{1,l}$ as the subsequent element of $A_{N,l}$). Thus, to sort the elements in the $l$th column of $\mathcal{A}$, we only need to find out the exception, i.e., the element that is larger than its subsequent element, which can be found with complexity $O(N)$. Denote the exception element as $A_{p,l}$. Then the $l$th column of the matrix can be sorted in ascending order by moving elements $A_{p+1,l}, A_{p+2,l},...,A_{N,l}$ to be above (in front of) $A_{1,l}$.
We do the sorting for all $K+1$ columns with total complexity  $O(N(K+1))$. Now we have $K+1$ sorted columns. We can merge the $K+1$ sorted columns into one sorted array using the Min-Heap algorithm \cite{ref29}, with complexity $O\big(N(K+1)\log (K+1)\big)$.

As a summary, in our fast algorithm to calculate $h^*$ in all the $N(K+1)$ sectors, we need $N(2K+3)$ vector additions and a sorting algorithm with complexity being $O\big(N(K+1)\log(K+1)\big)$. We can see that both the number of vector additions and the sorting complexity are linear with $N$.

\section{An Interesting Insight}\label{sec:empty_region}
From Theorem 1 in the preceding section, we know that ${h}^*$ cannot overlap with a separation line. Actually, next we will show that around a separation line, there exists a region that $h^*$ cannot be within that region. We call this region an {\it empty region}.

For the $n$th element, consider the separation line(s) between $F_{n,i}$ and $F_{n,i+1}$.

Firstly, consider the case that the counterclockwise angle from $F_{n,i}$ to $F_{n,i+1}$ is less than $\pi$. Then the $n$th element has a separation line, denoted $S_{n,i}$, between $F_{n,i}$ and $F_{n,i+1}$. The empty region around $S_{n,i}$ is shown in Fig.~\ref{lt90}. Next, we give an expression of the angle $\omega_{n,i}$.

We revisit the proof that $h^*$ cannot overlap with separation line $S_{n,i}$. This is because when $g_n$ takes either $F_{n,i}$ or $F_{n,i+1}$, there is always a contradiction. For example, assume $h^*$ overlaps with $S_{n,i}$ and $g_n$ takes $F_{n,i}$, as shown in Fig.~\ref{deltaf}. If $g_n$ switches from $F_{n,i}$ to $F_{n,i+1}$, then the achieved $h$ is expressed as $h^\dag = h^* + \Delta F_{n,i}$, in which $\Delta F_{n,i}=F_{n,i+1} - F_{n, i}$, as shown in Fig.~\ref{deltaf}. We can see that $h^*$, $\Delta F_{n,i}$ and $h^\dag$ form a triangle. Obviously we have $|h^*| < |h^\dag|$, which is a contradiction. Note that $\Delta F_{n,i}$ is a constant vector here. Now in Fig.~\ref{deltaf}, if $h^*$ rotates clockwise, then we can see that $h^\dag$ also rotates clockwise, and $|h^\dag|$ decreases. When $h^*$ keeps rotating, the contradiction always exists (i.e., we always have $|h^*| < |h^\dag|$) until $h^*$ arrives at a position such that $h^\dag$ and $h^*$ have equal amplitude, as shown in Fig.~\ref{omega}, Then the angle between $h^*$ and $S_{n,i}$ is $\omega_{n,i}$. Accordingly, we have
\begin{equation}\label{e:sin_omega}
\sin \omega_{n,i} = \frac{|\Delta F_{n,i}|/2}{|h^*|}.
\end{equation}
For vector $\Delta F_{n,i}$, from Fig.~\ref{omega} we have
\begin{equation}
|\Delta F_{n,i}| = 2 |F_{n,i}|\times |\sin \frac{\phi_{i+1}-\phi_i}{2}| = 2 |v_n|\times |\sin \frac{\phi_{i+1}-\phi_i}{2}|.
\end{equation}
Thus, we can get
\begin{equation}
\omega_{n,i} = \arcsin \frac{|v_n|\times |\sin \frac{\phi_{i+1}-\phi_i}{2}|}{|h^*|}.
\label{mtwidth1}
\end{equation}
\begin{figure}[!t]
\centering
\includegraphics[width=2.4in]{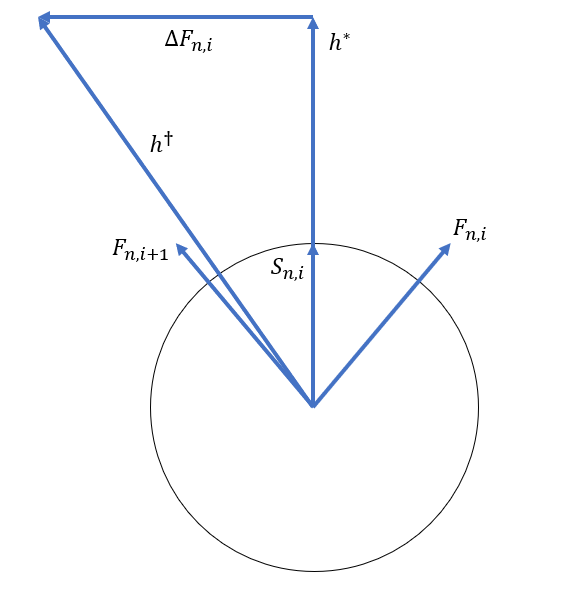}
\caption{$h^*$ overlapping with $S_{n,i}$.}
\label{deltaf}
\end{figure}

\begin{figure}[!t]
\centering
\includegraphics[width=2.4in]{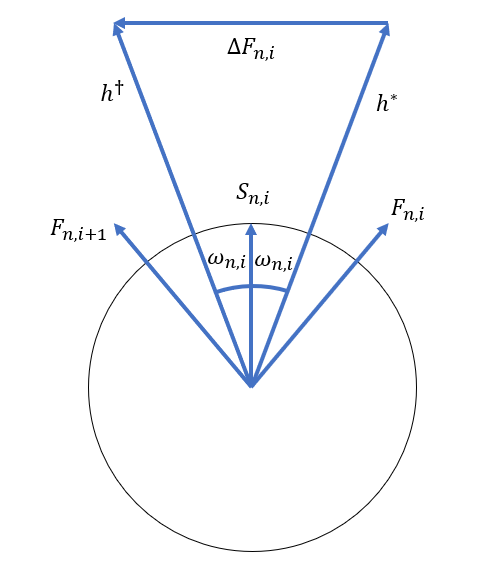}
\caption{$h^\dag$ and $h^*$ having equal amplitudes.}
\label{omega}
\end{figure}
Secondly, consider the case that the counterclockwise angle from $F_{n,i}$ to $F_{n,i+1}$ is more than $\pi$. Then the $n$th element has two separation lines, denoted as $S_{n,i,1}$ and $S_{n,i,2}$, between $F_{n,i}$ and $F_{n,i+1}$, as shown in Fig.~\ref{mtpi}. There is an empty region around each separation line, and the angle $\omega_{n,i,1}$ and $\omega_{n,i,2}$ are given as follows (the detailed derivations are omitted for presentation conciseness).
\begin{equation}
\omega_{n,i,1} = \omega_{n,i,2} = \arcsin \frac{|v_n|}{2|h^*|}.
\label{mtwidth2}
\end{equation}
\begin{figure}[!t]
\centering
\includegraphics[width=3in]{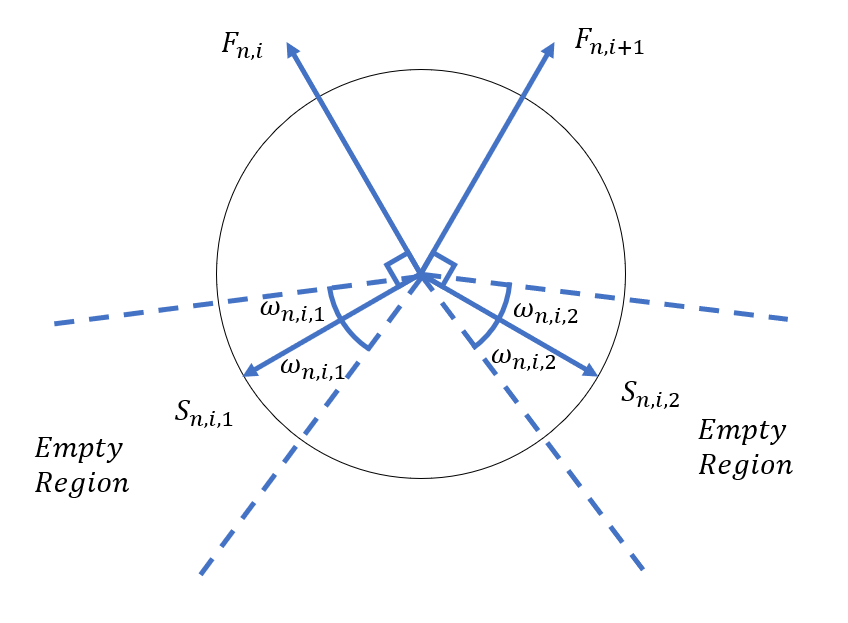}
\caption{Empty regions when the counterclockwise angle from $F_{n,i}$ to $F_{n,i+1}$ is more than $\pi$.}
\label{mtpi}
\end{figure}

Define {\it empty ratio} of the system  as the ratio of the total area of the empty regions to area of the whole circle. Also define $\gamma_\text{empty}^\text{UB}$ as the ratio of the summation of angles of all empty regions to $2\pi$. Since there may be overlap among the empty regions, we can see that $\gamma_\text{empty}^\text{UB}$ is an upper bound of the empty ratio, and superscript ``UB" stands for ``upper bound."

Next we give an approximation of $\gamma_\text{empty}^\text{UB}$ when $N$ is large.

For simplicity of approximation, we assume that for any $i$, the counterclockwise angle from $F_{n,i}$ to $F_{n,i+1}$ is always not more than $\pi$. Thus, the $n$th element has $K$ separation lines: $S_{n,1}, S_{n,2},...,S_{n,K}$, and the angle of the empty region for separation line $S_{n,i}$ is $2\omega_{n,i}$, and we have (\ref{mtwidth1}) for $\omega_{n,i}$, which is equivalent to the following equation.
\begin{equation}
\sin \omega_{n,i} = \frac{|v_n|\times |\sin \frac{\phi_{i+1}-\phi_i}{2}|}{|h^*|}.
\label{mtwidth1_new}
\end{equation}

When $N$ is large, $h^*$ is also large, and thus, the right-hand side of (\ref{mtwidth1_new}) is close to zero. Thus, expression (\ref{mtwidth1_new}) can be approximated as
\begin{equation}\label{e:apprx_omega_3}
\omega_{n,i} \approx \frac{|v_n|\times |\sin \frac{\phi_{i+1}-\phi_i}{2}|}{|h^*|}.
\end{equation}

$\phi_{i+1}-\phi_i$ is the gap between two phase shifts, with average value being $\frac{2\pi}{K}$. We use this average value to approximate $\phi_{i+1}-\phi_i$, and thus, expression  (\ref{e:apprx_omega_3}) becomes
\begin{equation}
\omega_{n,i} \approx \frac{ |v_n| \sin \frac{\pi}{K}}{|h^*|}.
\end{equation}
So for the $n$th element, the summation of the angles of empty regions associated with the element's $K$ separation lines (i.e., $S_{n,1},S_{n,2},...,S_{n,K}$) is expressed as
\begin{equation}
\sum_{i=1}^K 2\omega_{n,i} \approx \frac{2K|v_n| \sin \frac{\pi}{K} }{|h^*|}.
\end{equation}
Since we have $N$ elements with $NK$ separation lines, the summation of the angles of empty regions associated with the $NK$ separation lines is expressed as
\begin{equation}\label{e:apprx_omega_4}
\sum_{n=1}^N \sum_{i=1}^K 2\omega_{n,i} \approx \frac{ 2K \sin \frac{\pi}{K} \sum_{n=1}^N |v_n|}{|h^*|}.
\end{equation}
An upper bound of $|h^*|$ is $\sum_{n=1}^N |v_n|$ (i.e., when all RIS paths are perfectly aligned). We use this upper bound to approximate $|h^*|$ in (\ref{e:apprx_omega_4}) and get
\begin{equation}\label{e:apprx_omega_5}
\sum_{n=1}^N \sum_{i=1}^K 2\omega_{n,i} \approx 2K \sin \frac{\pi}{K}.
\end{equation}
Accordingly, when $N$ is large, $\gamma_\text{empty}^\text{UB}$ (upper bound of the empty ratio) is approximated as
\begin{equation}\label{e:empty_ratio_UB}
\gamma_\text{empty}^\text{UB} = \frac{\sum_{n=1}^N \sum_{i=1}^K 2\omega_{n,i}}{2\pi} \approx \frac{K\sin \frac{\pi}{K}}{\pi}.
\end{equation}

\section{Simulation Results}
This section shows simulation results that demonstrate the performance of the proposed method. In all simulations, the results are averaged over 1000 realizations. Path loss parameters are assumed to be the same for all RIS elements. $|h^\prime_n|$ and $|h^{\prime\prime}_n|$ are set to -80 dB and -60 dB, respectively. Also, $\angle(h_d)$ is set to 0 and $\angle(v_n)$ is chosen uniformly at random from the interval $[0,2\pi)$ for each element.

In order to make a comparison, an alternative method called \emph{Closest Point Projection (CPP)}\cite{ref11} is utilized to solve the problem. This method assumes that $h^{*}$ is in the same direction as $h_d$ and configures the elements in a way that they are as close as possible to $h_d$.  The results of the exhaustive search over all $(K+1)^N$ possible configurations are also presented to verify the optimality of the proposed method.

Fig.~\ref{capvsn} shows the channel capacity ($C$) of the proposed method, the CPP method, and the exhaustive search method. $|h_d|$ is set to -140 dB and $\frac{P}{BN_0}$ is set to 100 dB.
An arbitrary set of phase shifts has been chosen for each element:
\begin{equation}
\theta_n \in \{e^{j\dfrac{\pi}{6}}, e^{j\dfrac{5\pi}{6}}, 0{\}}.
\end{equation}
It is seen that the proposed method has exactly the same channel capacity as that of the exhaustive search method and is always better than the CPP method. This verifies the optimality of our proposed method.

\begin{figure}[!t]
\centering
\includegraphics[width=3.8in]{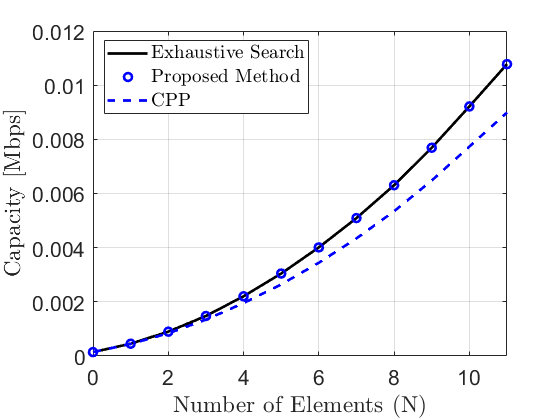}
\caption{Channel capacity of the proposed method, the CPP method, and the exhaustive search method.}
\label{capvsn}
\end{figure}


In the subsequent simulation, we will measure only the amount of performance gain that the proposed method gives us over the CPP method. This gain is defined as:
\begin{equation}
\text{Performance\;Gain\;[\%]} = \frac{C_{\text{Proposed}} - C_{\text{CPP}}}{C_{\text{CPP}}}\times 100,
\end{equation}
in which $C_{\text{Proposed}}$ and $C_{\text{CPP}}$ are channel capacity achieved by the proposed method and the CPP method, respectively.

Fig.~\ref{pgvshd} shows the performance gain when the $|h_d|$ varies from -140~dB to -100~dB. $N$ is set to 50 and the rest of the parameters are similar to those in Fig.~\ref{capvsn}. When $|h_d|$ is small, i.e., there is a weak direct path between the transmitter and receiver, the performance gain is high. The gain becomes lower and converges to 0 as $|h_d|$ becomes larger. The reason is as follows. If $|h_d|$ is large enough, $h_d$ would become the dominant term in \eqref{deqn_ex55} and thus, both $C_{\text{Proposed}}$ and $C_{\text{CPP}}$ would be almost equal to $B \log_2(1+\frac{P|h_d|^2}{BN_0})$.
\begin{figure}[!t]
\centering
\includegraphics[width=3.6in]{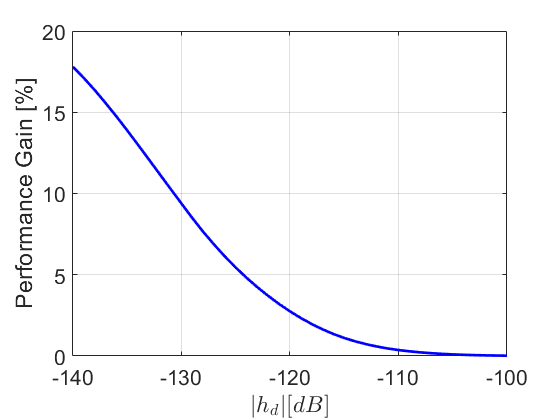}
\caption{Performance gain versus $|h_d|$.}
\label{pgvshd}
\end{figure}

\begin{figure}[!t]
\centering
\includegraphics[width=3.6in]{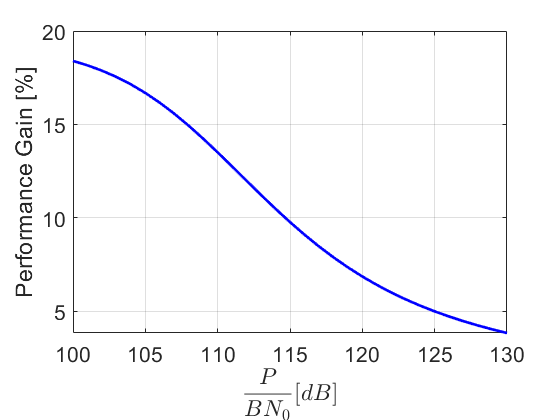}
\caption{Performance gain versus $\frac{P}{BN_0}$.}
\label{pgvspbn}
\end{figure}

Fig.~\ref{pgvspbn} shows the performance gain when $\frac{P}{BN_0}$ varies. $N$ is set to 50 and the rest of the parameters are similar to those of  Fig.~\ref{capvsn}.
It is seen that with the same $\frac{P}{BN_0}$, the proposed method always has a performance gain (i.e., the proposed method can result in a higher capacity than CPP). For example, when
$\frac {P}{BN_0}$ is 108 dB, the proposed method provides us 15 percent more capacity compared to CPP. In other words, the proposed method is more energy-efficient than CPP. It is also seen that the performance gain is higher when  $\frac{P}{BN_0}$ is small and decreases as $\frac{P}{BN_0}$ becomes larger. This is because when $\frac{P}{BN_0}$ is large enough, both $C_{\text{Proposed}}$ and $C_{\text{CPP}}$ would be almost equal to $B \log_2(\frac{P}{BN_0})$ and thus, the performance gain converges to 0.

Next, we will study how the choice of available phase shifts can affect the performance gain. When there are only two phase shifts, we can try different possibilities simply by changing the phase gap between the two phase shifts. This can be seen in Fig.~\ref{pgvsgap}. $N$ is set to 50 and the rest of the parameters are similar to those of Fig.~\ref{capvsn}. In all cases, significant gains are reported, with higher gains when the two available phase shifts deviate more from a uniform distribution. As we deviate from this uniform distribution point (i.e., when the phase gap is $\pi$), the gain significantly increases.

\begin{figure}[!t]
\centering
\includegraphics[width=3.8in]{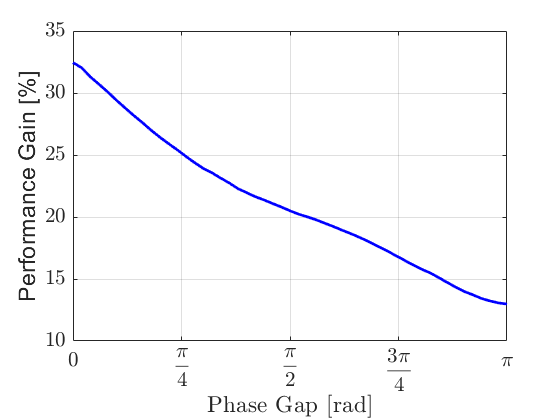}
\caption{Performance gain versus phase pap when there are two phase shifts.}
\label{pgvsgap}
\end{figure}

Now, let us try the same simulation again but with a larger set of configurations. When our set of configurations consists of three phase shifts, we can go through all possibilities by defining two phase gaps: One between the first and second phase shifts (phase gap 1) and the other one between the second and third phase shifts (phase gap 2). The result can be seen in Fig.~\ref{pgvsgap2}. $N$ is set to 50 and the rest of the parameters are similar to those of Fig.~\ref{capvsn}. Again, we can see the performance gain is minimum when both phase gaps are around $\frac{2\pi}{3}$ which is equivalent to the uniform distribution. When we deviate from this uniform distribution point, performance gain largely increases.
\begin{figure}[!t]
\centering
\includegraphics[width=3.4in]{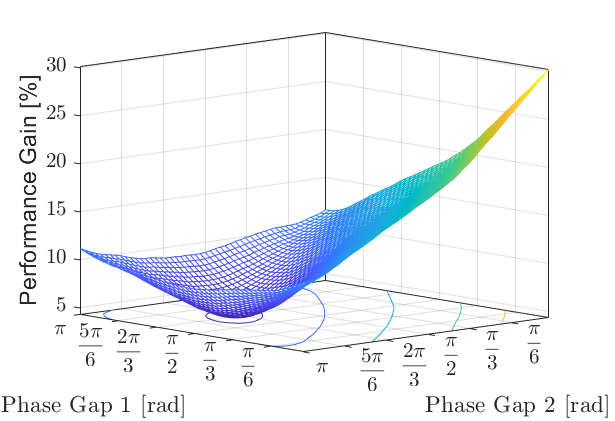}
\caption{Performance gain versus phase gap 1 and phase gap 2 when there are three phase shifts.}
\label{pgvsgap2}
\end{figure}

In the end, we are going to show how much portion is covered by empty regions. Here we assume $|h^*|$ is known and use \eqref{mtwidth1}, \eqref{mtwidth2} to calculate the angle $\omega$ for the separation lines. A single realization can be seen in Fig.~\ref{mtrealiz}. $N$ is set to 50 and the rest of the parameters are similar to those of Fig.~\ref{capvsn}. We expect to see $N(K+1) = 150$ separation lines. The black lines are the separation lines and the colored area covering each line is the empty region. As we can see, the majority of the circle is covered by empty regions.
\begin{figure}[!t]
\centering
\includegraphics[width=4in]{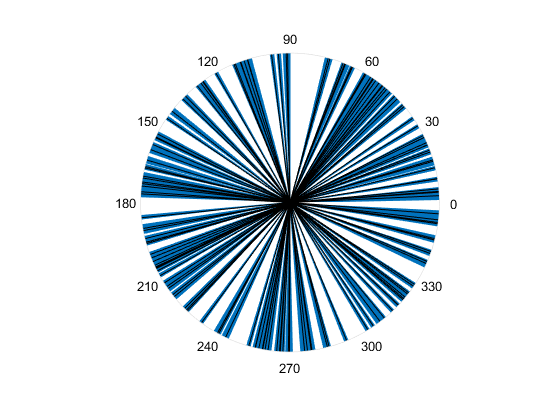}
\caption{A single realization demonstrating empty regions around each separation line (the numbers outside the circle have unit ``degree").}
\label{mtrealiz}
\end{figure}

Fig.~\ref{mtratio} shows the empty ratio versus the number of elements. Phase shifts are chosen uniformly for both $K=2$ and $K=3$. The rest of the parameters are similar to those of Fig.~\ref{capvsn}. It is seen that when $N$ becomes larger, the empty ratio converges to $0.6$ for $K=2$ and $0.62$ for $K=3$. From (\ref{e:empty_ratio_UB}), $\gamma_\text{empty}^\text{UB}$ (upper bound of the empty ratio) is calculated as $0.64$ for $K=2$ and $0.83$ for $K=3$. The simulated empty ratio values indeed are within the upper bound values.

\begin{figure}[!t]
\centering
\includegraphics[width=3.8in]{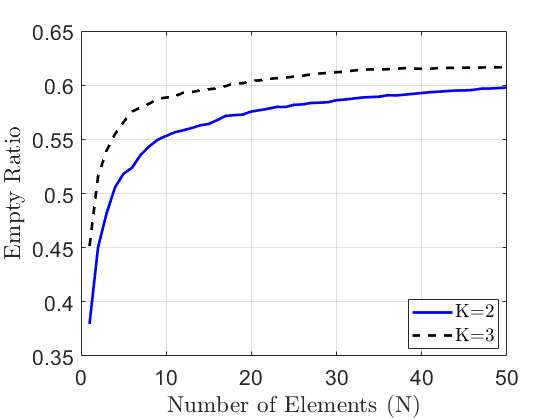}
\caption{Empty ratio versus the number of elements.}
\label{mtratio}
\end{figure}

\section{Conclusion}
For an RIS with discrete phase shifts, it is more reasonable to assume non-uniform phase shifts. This paper has investigated the optimal configuration of arbitrary non-uniform phase shifts. We have demonstrated that turning on all elements, which is commonly adopted in the literature, may not be optimal. We have theoretically proved that each element should be turned on with the highest reflection amplitude or be simply turned off. We have proposed a method that employs a series of search steps to determine the optimal configuration for each element. Notably, the number of search steps in our proposed method scales linearly with the number of elements. We have also proposed a fast algorithm to further reduce the computations in the search steps. We have also demonstrated the existence of empty regions, and our simulation results have shown that the empty regions occupy a large portion of the circle in the complex plane.


\bibliographystyle{IEEEtran}
\bibliography{biblo.bib}

\end{document}